\documentclass[conference,a4paper]{IEEEtran}

\usepackage[utf8]{inputenc} 
\usepackage[T1]{fontenc}
\usepackage{url}
\usepackage{ifthen}
\usepackage{cite}
\usepackage[cmex10]{amsmath} 

\usepackage{amsthm,amssymb}
\usepackage{graphicx}
\usepackage{amsfonts}
\usepackage{color}
\usepackage{bm}
\usepackage{hyperref}

\usepackage{makecell}
 \newcounter{thm}
  \newcounter{re}
\newtheorem{lemma}{Lemma}
\newtheorem{assumption}{Assumption}

\usepackage{algorithmic}
\usepackage{algorithm}

\DeclareMathOperator*{\argmax}{arg\,max}

\begin{document}
\title{Asynchronous Multi-Sensor Change-Point Detection for Seismic Tremors}



 \author{%
   \IEEEauthorblockN{Liyan Xie\IEEEauthorrefmark{1},
                      Yao Xie\IEEEauthorrefmark{1},
                     George V. Moustakides \IEEEauthorrefmark{2}\IEEEauthorrefmark{3}
                     }
   \IEEEauthorblockA{\IEEEauthorrefmark{1}%
                     School of Industrial \& Systems Engineering, Georgia Institute of Technology,
                     \{lxie49, yao.xie\}@isye.gatech.edu }
   \IEEEauthorblockA{\IEEEauthorrefmark{2}%
                     Department of Computer Science, Rutgers University, gm463@rutgers.edu
                     }
   \IEEEauthorblockA{\IEEEauthorrefmark{3}%
                     Electrical and Computer Engineering department, University of Patras, moustaki@upatras.gr }
   }

\maketitle


\begin{abstract}
THIS PAPER IS ELIGIBLE FOR THE STUDENT PAPER AWARD.  
 We consider the sequential change-point detection for asynchronous multi-sensors, where each sensor  observe a signal (due to change-point) at different times. We propose an asynchronous Subspace-CUSUM procedure based on jointly estimating the unknown signal waveform and the unknown relative delays between the sensors. Using the estimated delays, we can align signals and use the subspace to combine the multiple sensor observations. We derive the optimal drift parameter for the proposed procedure, and characterize the relationship between the expected detection delay, average run length (of false alarms), and the energy of the time-varying signal. We demonstrate the good performance of the proposed procedure using simulation and real data. We also demonstrate that the proposed procedure outperforms the well-known ``one-shot procedure'' in detecting weak and asynchronous signals. 
\end{abstract}


\section{Introduction}\label{sec:intro}

We consider detecting the emergence of a signal, which is observed at multiple sensors with {\it unknown and different delays and amplitudes}. Such problem arises frequently in sensor network monitoring, where the sensors observe the sudden occurrence of a signal at different times mostly due to propagation delays, as illustrated in Fig. \ref{fig:illustration}. The main application of interest is seismic sensors for detecting tremors \cite{lin2008surface}. Tremors are low amplitude ambient vibrations of the ground caused by man-made or atmospheric disturbances; detecting the underlying tremors will enable geophysicists to build better predictive models.  Our goal is to detect the emergence of such occurrence (change) by combining the sensor observations.
\begin{figure}[h!]
\vskip-0.3cm
\centerline{
\includegraphics[width=0.4\textwidth]{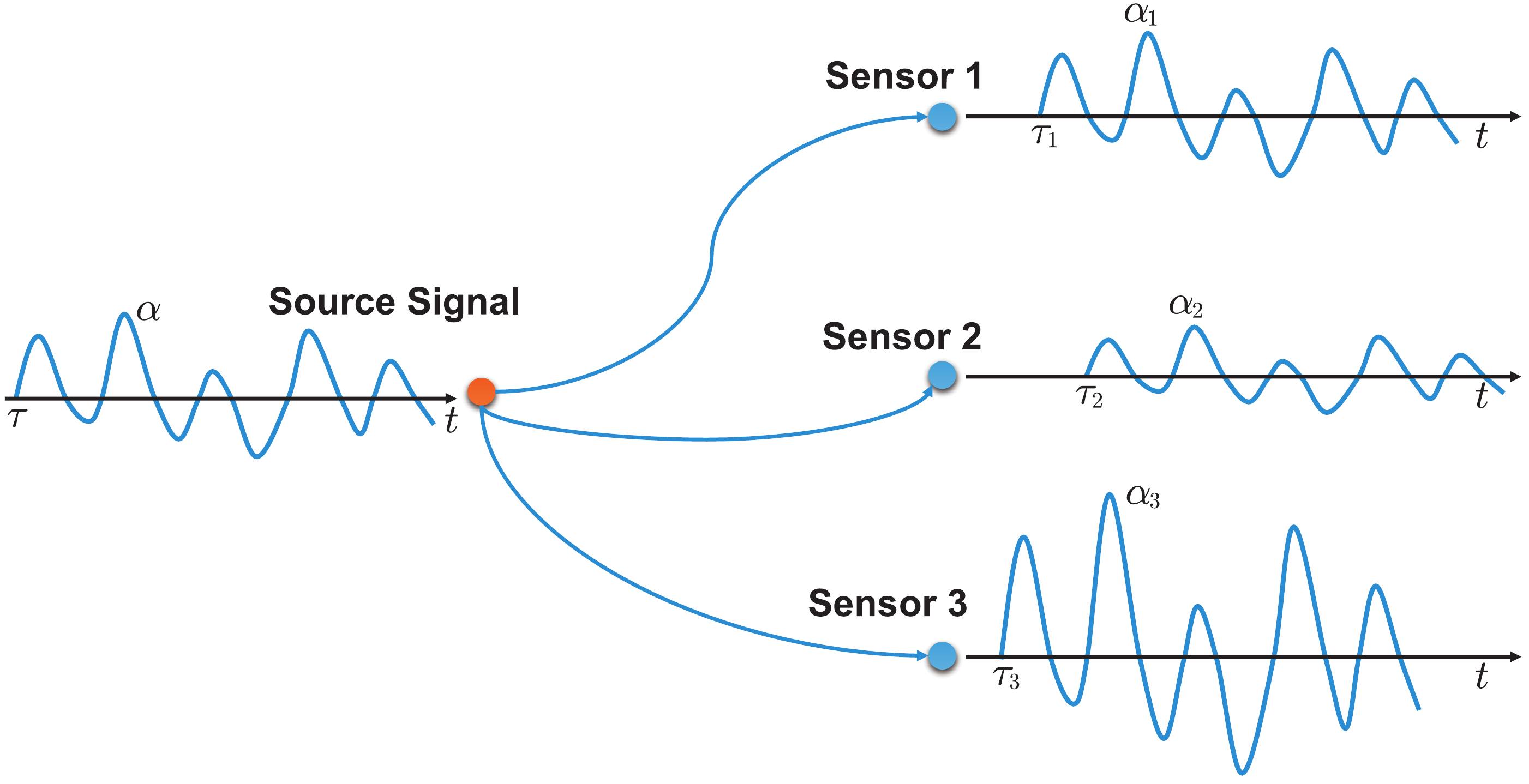}  
}
\vskip-0.3cm
\caption{Illustration of a sensor network with a signal.}
\label{fig:illustration}
\vskip-0.3cm
\end{figure}

Classical sequential change-point detection \cite{tartakovsky2014sequential} usually assumes the change happens simultaneously at all or a subset of sensors. The scenario we are proposing above, calls for the development of new methods that can consider delay estimation together with change-point detection. Change-point detection with delays has also been considered in \cite{hadjiliadis2009one} using the so-called one-shot scheme, where the fusion center declares an alarm whenever one sensor detects a change. However, the one-shot scheme relies only on local sensor information, it does not take advantage of the multi-sensor observations and the fact that the change occurs in all sensors but not at the same time. Combining asynchronous sensor observations can be beneficial for the detection problem when the relative delays between sensors are not large and when the signal is weak, as is the case of seismic tremors. This is because combining observations effectively boosts the SNR.

Even though in our work the time of change in sensors appears to be different, there exists a \textit{deterministic difference} between change-times due to the location geometry of sensors and the location where the tremor occurs. Consequently, our approach will be similar to approaches where there is a simultaneous change in all sensors after, of course, compensating for the fixed but unknown delays.  

We propose an asynchronous Subspace-CUSUM  procedure, based on jointly estimating the unknown signal waveform and the unknown relative delays. It is related to the Subspace-CUSUM procedure in our prior work \cite{xie2018first, xie2018sequential}. We extend the results therein for the asynchronous case, and develop an optimal choice  of the drift parameter, which is essential for CUSUM type of procedures. Our theoretical analysis reveals insights into the relationship between the average energy of the time-varying signal and the expected detection delay. This may potentially allow us to prove the asymptotic optimality of the asynchronous Subspace-CUSUM, by  extending the arguments in \cite{xie2018first}. We demonstrate the good performance of our procedure using simulated examples and real seismic data. Our procedure outperforms the one-shot scheme especially when the signal amplitude is weak and the relative delays are not too large.  

The rest of the paper is organized as follows. In section \ref{sec:background}, we introduce the background of sequential subspace detection and briefly summarize the Subspace-CUSUM procedure. In section \ref{sec:formulation}, we first consider the case where the delays are perfectly known and form the detection statistics. Next, we propose a detector that combines sensor synchronization with detection when the delays are unknown. The theoretical result about how to properly select the drift parameter is discussed in section \ref{sec:theoretical}. In section \ref{sec:numerical} we evaluate our method by applying it to both simulated signals and real seismic data and observe its performance under strong and weak signals.

\section{Background: Subspace-CUSUM procedure}\label{sec:background}

We first introduce the Subspace-CUSUM procedure which will be the basis of our subsequent discussion. We consider the change-point detection problem where the covariance changes from an identity matrix $\sigma^2 I_k$ to a spiked matrix $\Sigma = \sigma^2I_k + \theta uu^\intercal$, where $\theta > 0$ is the signal strength, $u \in \mathbb{R}^{k \times 1}$ represents a basis for the subspace with $\|u\|_2=1$, and where $\sigma^2$ is the noise power. We can define the SNR as $\rho = \theta/\sigma^2$. Assume the sequentially observed data are as follows
\begin{equation}
\begin{array}{ll}
x_t  \stackrel{\text{iid}}{\sim} \mathcal{N}(0,\sigma^2I_k ), &t = 1,2,\ldots,\tau, \\
x_t \stackrel{\text{iid}}{\sim} \mathcal{N}(0,  \sigma^2I_k + \theta u u^\intercal),& t = \tau+1,\tau+2, \ldots
\end{array}
\label{eq:hypothesis}
\end{equation}
where $\tau$ is the unknown change-point that we would like to detect as soon as possible. We assume that the subspace $u$ is unknown since it represents anomaly or new information. 

The well-known cumulative sum (CUSUM) test \cite{page1954continuous,moustakides1986CUSUM} cumulates the \textit{log-likelihood ratio} and declares an alarm whenever the cumulation exceeds a threshold. For the data model in \eqref{eq:hypothesis}, we can derive the log-likelihood ratio for each sample as the equation (7) in \cite{xie2018sequential}:
\[
\log\frac{f_0(x_t)}{f_\infty(x_t)} = \frac1{2\sigma^2}\frac{\rho}{1+\rho} \Big\{(u^\intercal x_t)^2- \sigma^2 \!\left(\!1+\frac{1}{\rho}\right)\!\log(1+\rho)\Big\},
\]
where $f_\infty$ and $f_0$ denote the probability density function before and after the change. Based on this, we can form the CUSUM statistic as
\begin{equation*}
S_t = (S_{t-1})^+ + (u^\intercal x_t)^2 - \sigma^2\left(1+\frac{1}{\rho}\right)\log(1+\rho),
\label{cusum_recur1}
\end{equation*}
where $(x)^+=\max\{x, 0\}$. When $u$ is unknown, in the Subspace-CUSUM procedure, $u$ is replaced with a sequential estimate $\hat{u}_t$:
\begin{equation}
\mathcal S_t = (\mathcal S_{t-1})^{+} + (\hat{u}_t^\intercal x_t)^2 - d.
\label{sscusum_update}
\end{equation}
Here, the parameter $d$ is the \textit{drift} parameter that we would like to select properly so that the increment of $\mathcal S_t$ has a negative mean under the nominal and a positive mean under the alternative probability measure. This requires
\begin{equation}\label{eq:d_interval}
\mathbb{E}_\infty[(\hat{u}_t^\intercal x_t)^2] < d < \mathbb{E}_0[(\hat{u}_t^\intercal x_t)^2],
\end{equation}
where $\mathbb{E}_\infty$ and $\mathbb{E}_0$ denote the expectation under nominal and alternative measure respectively. The Subspace-CUSUM procedure can be defined through the following stopping time
\begin{equation}
\mathcal T = \inf\{t>0: \mathcal S_t \geq b\},
\label{cusum_procedure}
\end{equation}
where $b$ is a pre-specified threshold set to control the false alarm rate. This test is known to be asymptotically optimum for the stationary case (constant $\theta$) \cite{xie2018first}.

To obtain the estimate $\hat u_t$, we form the sample covariance matrix using the observations $\{x_{t+1}, \ldots,x_{t+w}\}$ that lie in the future of $t$, 
\begin{equation}
\Sigma_t = \textstyle\sum\limits_{j = t+1}^{t+w} x_j x_j^\intercal,
\label{eq:Sigma}
\end{equation}
then the {\it unit-norm} singular vector corresponding to the largest singular value of $\Sigma_t$ can be viewed as an estimator for $u$ at time $t$. 
The usage of observations from the future is always possible by properly delaying the data. In particular, if we stop at time $\mathcal T = t$, this implies that we used data from times up to $t + w$ and, consequently, $t + w$ is the true time we stop and not $t$. The main advantage of this idea is that it provides the estimator $\hat u_t$ that is independent of $x_t$.

\section{Problem Setup} \label{sec:formulation}

In this section, we first consider the case when delays are known, and show how the Subspace-CUSUM procedure should be applied to detect the signal. When the delays are unknown, we develop a method that simultaneously synchronizes sensors by estimating their relative delays and detects the change.

Consider $k$ sensors as in Fig.\,\ref{fig:illustration}. Suppose we have sequential observations at each sensor, $x_1(t), x_2(t), \ldots,x_k(t)$. Before the emergence of the signal, the observations are noises that are assumed to be {\it i.i.d.} normal random variables. We will also assume that their powers are known which implies that, without loss of generality, we can assume that these powers are all equal to $\sigma^2$ since this is always possible with proper normalization. When a source signal $s(t)$ occurs, the observations at different sensors will capture it but with different delays and amplitudes. We assume the signal is  causal, i.e., $s(t) = 0, \forall t < 0$. Denote the time point of onset of the signal in $i$th sensor as $\tau_i$, and define the change-point as 
\[
\tau = \min_{1\leq i \leq k} \tau_i.
\]
For the $i$th sensor, this means:
\begin{equation}
\begin{array}{ll}
x_i (t) = e_i(t), &  t = 1,2,\ldots, \tau , \\
x_i(t)  = \alpha_i s(t - \tau_i) + e_i(t), & t = \tau+1, \tau +2, \ldots,
\end{array}
\label{eq:model}
\end{equation}
where $e_i(t) \stackrel{\text{iid}}{\sim} \mathcal N(0,\sigma^2)$ are random noises, and $\alpha_i$ is the {\it unknown} amplitude of the change at the $i$th sensor. Assume that the source signal $s(t)$ is {\it unknown}. We also define the relative delay between the $i$th and $j$th sensor as $\tau_{ij} = \tau_j - \tau_i$. Note that the delays with respect to the source signal are always nonnegative, but the relative delays can be either positive or negative.

\subsection{Known Delays}\label{sec:known}

To help build intuition, we first consider the ideal case where the relative delays $\tau_{ij}$ are known. Without loss of generality, we assume $\tau_{1i} \geq 0, \forall i $. If this is the case we can construct the time-shifted version of the observations
\begin{equation}
\begin{array}{ll}
\tilde x_i(t)  = x_i(t+\tau_{1i}) ,  & i = 1,2,\ldots,k,
\end{array}
\label{eq:model_synch}
\end{equation}
and combine them to form a $k$-dimensional vector
\begin{equation}\label{eq:sample}
\tilde{x}_t = \begin{bmatrix}
\tilde x_1(t) & \tilde x_2(t) & \cdots & \tilde x_k(t) 
\end{bmatrix}^\intercal.
\end{equation}
Note that when there is a signal, after the data transformation in \eqref{eq:model_synch}, we can show that the vectorized observations can be written as  
\begin{equation}
\tilde{x}_t = s(t) \begin{bmatrix}
\alpha_1 & \alpha_2 & \cdots &  \alpha_k \end{bmatrix}^\intercal + e_t, 
\,  e(t) \sim \mathcal N(0,\sigma^2 I_k). \label{model1}
 \end{equation}
Therefore the covariance structure of  $\tilde{x}_t$ will undergo a similar change as in the model in \eqref{eq:hypothesis}, namely, be the identity matrix $\sigma^2I_k$ before the emergence of the signal $s(t)$, and become the spiked model $\sigma^2I_k + \theta(t) u u^\intercal$ after. 
Here the subspace $u$, represents the unknown  {\it normalized} post-change amplitudes
\begin{equation}\label{eq:subspace}
u = \frac{\alpha}{\|\alpha\|},
\end{equation}
where, for simplicity, we denote $\alpha=[\alpha_1,\ldots,\alpha_k]^\intercal$, and the {\it time-varying} signal strength $\theta(t)$ is given by
\[
\theta(t)  = s^2(t)\| \alpha \|^2.
\]
Note that (\ref{model1}) differs from the previous model \eqref{eq:hypothesis} in that $\theta(t)$ is now \textit{time-varying}, since the signal $s(t)$ changes with time. However, the change-point detection is still to detect the transition of the covariance matrix from an identity to a spiked covariance matrix. So, we can still adapt the Subspace-CUSUM in section \ref{sec:background} to be used in this case. We should add that optimum detection schemes for time-varying models are difficult to derive. For existing results please refer to \cite{tartakovsky2014sequential}.

Since the normalized post-change amplitude vector $u$ is unknown, we can estimate it by forming the sample covariance matrix $\Sigma_t$ introduced in \eqref{eq:Sigma} but with $x_t$ replaced by $\tilde{x}_t$ in \eqref{eq:sample}.
If we apply the singular value decomposition on $\Sigma_t$ then the singular vector corresponding to the largest singular value provides the desired estimate $\hat{u}_t$. Then, we plug $\hat u_t$ back into the detection statistic \eqref{sscusum_update} and obtain
\begin{equation}\label{eq:detect_stat}
\mathcal S_t = (\mathcal S_{t-1})^{+} + (\hat{u}_t^\intercal {\tilde x_t})^2 - d,
\end{equation}
and the stopping time is defined as in \eqref{cusum_procedure}. Of course, there still remains the question of selecting the proper $d$. We defer the discussion of this issue until Section\,\ref{sec:theoretical}.

\subsection{Unknown Delays}\label{sec:unknown}
In this section, we consider the case where delays are unknown. Generally, the exact delay is not possible to obtain beforehand since it depends on the location of the tremor epicenter. Therefore, we need to come up with a method that will achieve sensor synchronization in order to apply the Subspace-CUSUM procedure. In fact, this will be performed continuously and in a sequential manner in parallel with the change-detection task. 

We select one sensor as reference, and attempt to synchronize all other sensors with respect to this sensor. Synchronization can be implemented based on the maximum likelihood approach to estimate the relative delay, on a sensor-by-sensor basis or simultaneously for all sensors. Without loss of generality, we regard the data of the first sensor $x_1(t)$ as the reference and compute the relative delays with respect to $x_1(t)$. Assume that we have available an upper bound $\tau_{\max}$ on the unknown relative delays, so they are restricted in the interval $[-\tau_{\max}, \tau_{\max}]$. 

For the $i$th sensor, the log-likelihood function of the observations $\{x_i(t+1),\ldots, x_i(t+w)\}$ after change can be written as
\begin{multline*}
\ell_{s,\tau_i}\big(x_i(t+1),\ldots,x_i(t+w) \big) = \\
  \frac{\alpha_i }{\sigma^2}\sum_{j = t+1}^{t+w} x_i(j) s(j-\tau_i) -\frac{\alpha^2_i}{2\sigma^2}\sum_{j = t+1}^{t+w}s^2(j-\tau_i)  \\
 - \frac{1}{2\sigma^2} \sum_{j = t+1}^{t+w} x^2_i(j) - \frac w2 \log(2\pi\sigma^2). 
\end{multline*}
Therefore for any given signal waveform $s(t)$, the maximum likelihood estimator of $\tau_i$ at $i$th sensor is given by 
\begin{equation}\label{eq:delayMLE}
\hat \tau_i = \argmax_{-\tau_{\max} \leq z \leq \tau_{\max} } {\left|\sum\nolimits_{j = t+1}^{t+w} x_i(j) s(j-z)\right|}.
\end{equation}
 Based on the maximum likelihood estimator \eqref{eq:delayMLE}, we propose Algorithm \ref{alg1} which performs the joint estimation of signal waveform and relative delay iteratively. 
\begin{algorithm}[h]
\caption{Joint estimate of signal waveform and delay} \label{algo:estimate}
\begin{algorithmic}[1]
\REQUIRE $\delta$, $n_{\max}$
 \STATE \textbf{Initialize:}  $n \leftarrow 1$; $ \hat s^{(1)} \leftarrow x_1 $; $\hat \tau^{(0)}_{1i} = \infty$, $\hat \tau^{(1)}_{1i} = 0, \forall i$
 \WHILE{ $\max_{i\geq2} |\hat \tau_{1i}^{(n)} - \hat \tau_{1i}^{(n-1)} | \geq \delta$ and $n\leq n_{\max}$}
 \STATE $n \leftarrow n+1$
  \FOR {$i = 2,\ldots,k$ }
\STATE $ \hat{\tau}^{(n)}_{1i}=\argmax\limits_{-\tau_{\max}\leq z \leq\tau_{\max}} | \sum\nolimits_{j=t+1}^{t+w} x_i(j)\hat s^{(n-1)}(j-z)|$ 
\ENDFOR
\STATE Form sample vector \eqref{eq:sample} using delay estimate $\hat{\tau}^{(n)}_{1i}$
\STATE Find $\hat u$, the singular vector corresponding to the largest singular value of the sample covariance matrix \eqref{eq:Sigma}
 \STATE $\hat s^{(n)}(t) \leftarrow \sum_{i=1}^k \hat u_i x_i(t+\hat \tau_{1i}^{(n)}) $
 \ENDWHILE
\end{algorithmic}
\label{alg1}
\end{algorithm}

Once we obtain the estimates of the delays we can then use \eqref{eq:model_synch} with $\tau_{1i}$ replaced by $\hat{\tau}_{1i}$ and then apply the Subspace-CUSUM as described above with $x_t$ replaced by $\tilde{x}_t$. 

\section{Theoretical analysis}\label{sec:theoretical}

Our previously introduced condition \eqref{eq:d_interval} for the drift parameter is necessary to guarantee that Subspace-CUSUM will exhibit the same performance as the regular CUSUM. In other words, the increment will have a negative mean under the nominal model resulting in multiple restarts while under the alternative regime the increment will be positive on average leading the statistic $\mathcal{S}_t$ to exceed the threshold. The aforementioned property is crucial and it ensures that the detection delay is proportional to the threshold while the expected duration between false alarms is an exponential function of the threshold. 

When the post-change statistical behavior is not stationary, \eqref{eq:d_interval} is no longer applicable and we need to consider time as an additional source of variability. Since in our problem $s(t)$ is time-varying, the expectation over $s(t)$ must be replaced with the average \textit{over time}. Specifically we have
\begin{multline}
\lim_{W\to\infty}\frac{1}{W}\sum_{j=0}^{W-1}\mathbb{E}_\infty[(\hat{u}_{t+j}^\intercal x_{t+j})^2] < d\\
 < \lim_{W\to\infty}\frac{1}{W}\sum_{j=0}^{W-1}\mathbb{E}_0[(\hat{u}_{t+j}^\intercal x_{t+j})^2].
 \label{eq:d_interval2}
\end{multline}
In other words we assume that \eqref{eq:d_interval} is valid \textit{on average} over time. Let us now see how this translates in our specific problem. We first make the following assumption.
\begin{assumption}\label{ass:average_energy}
There exists a positive constant $E_0$ such that the following limit is valid
\begin{equation}\label{eq:average_energy}
\lim_{W \rightarrow \infty} \frac1W \sum_{i=1}^{W} s^2(i) = E_0>0.
\end{equation}
Quantity $E_0$ denotes the average energy of the signal $s(t)$.
\end{assumption}
To compute the expectations of $(\hat{u}_t^\intercal \tilde{x}_t)^2$, especially under the alternative regime, it is necessary to be able to describe the statistical behavior of our estimate $\hat{u}_t$. We will assume that the window size $w$ is sufficiently large so that Central Limit Theorem type approximations are possible. Explicit formulas are given in the following Lemma.
\begin{lemma}\label{lem:drift}
If Assumption\,\ref{ass:average_energy} is true, and $e_i(t) \stackrel{{\rm iid}}{\sim} \mathcal N(0,\sigma^2)$, we have that under the pre-change regime,
\[
\mathbb{E}_\infty \left[  (\hat{u}_t^\intercal \tilde{x}_t)^2 \right] = \sigma^2, 
\]
and under the post-change regime:
\begin{multline*}\label{eq:drift0}
\mathbb{E}_0 \left[  (\hat{u}_t^\intercal \tilde{x}_t)^2 \right] 
= \sigma^2 \left[ 1 + \frac{s^2(t)\rho}{E_0}\left( 1 -  \frac{1+\rho}{w\rho^2}(k-1)  \right) \right],
\end{multline*}
where
\begin{equation*}\label{eq:rho_definition}
\rho = \frac{E_0}{\sigma^2} \| \alpha\|^2,
\end{equation*}
can be viewed as the average SNR over time.
\end{lemma}

\begin{proof} The proof is given in the Appendix.
\end{proof}

Using Lemma\,\ref{lem:drift} and \eqref{eq:d_interval2} we can immediately deduce that the drift $d$ must satisfy
$$
\sigma^2<d< \sigma^2 \left[ 1 + \rho\left( 1 -  \frac{1+\rho}{w\rho^2}(k-1)  \right) \right],
$$
which is similar to the stationary condition imposed in \cite{xie2018first} but with the average energy $E_0$ replacing the constant energy of the stationary version.

\section{Numerical results}\label{sec:numerical}

\subsection{Simulation}
 
In this section, we perform simulations to compare the performance of our subspace-based test with the one-shot scheme proposed in \cite{hadjiliadis2009one}. We adopt the same setting as in \cite{hadjiliadis2009one}, where the distribution of the data stream at each sensor changes from $\mathcal N(0,\sigma^2)$ to $\mathcal N(\mu, \sigma^2)$ asynchronously. Indeed this is a special case of the model \eqref{eq:model} by letting the signal $s(t) = 1$, and the amplitude $\alpha_i = \mu $ for all sensors. In the numerical experiments, the noise level $\sigma^2 = 1$, the number of sensors $k=50$, the maximal relative delay $\tau_{\max}$ is set to $20$, and the window size $w$ is set to $20$.

\begin{figure}[h]
\vskip-0.3cm
\centerline{
\includegraphics[width = 0.23\textwidth]{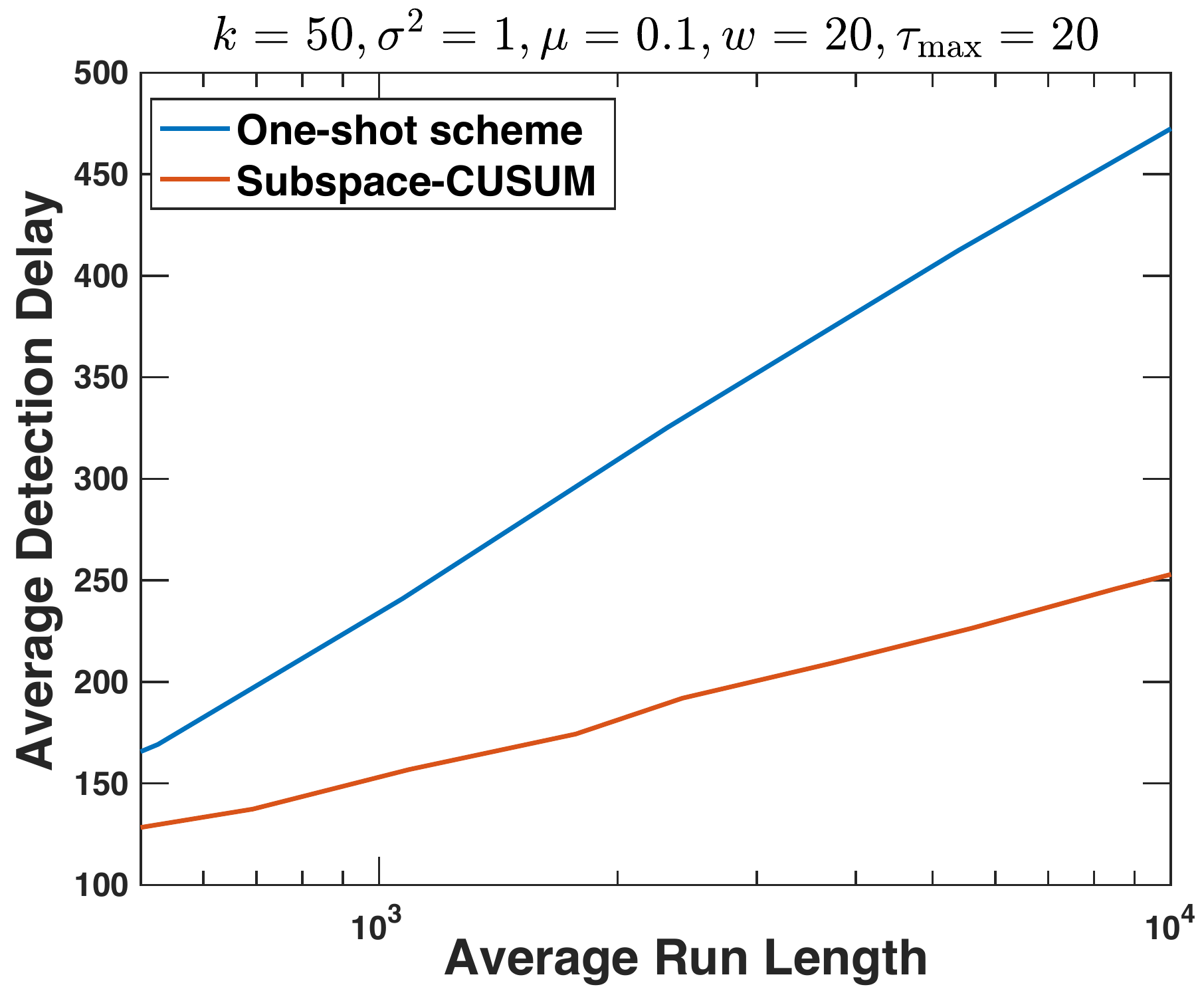} \hfill \includegraphics[width = 0.23\textwidth]{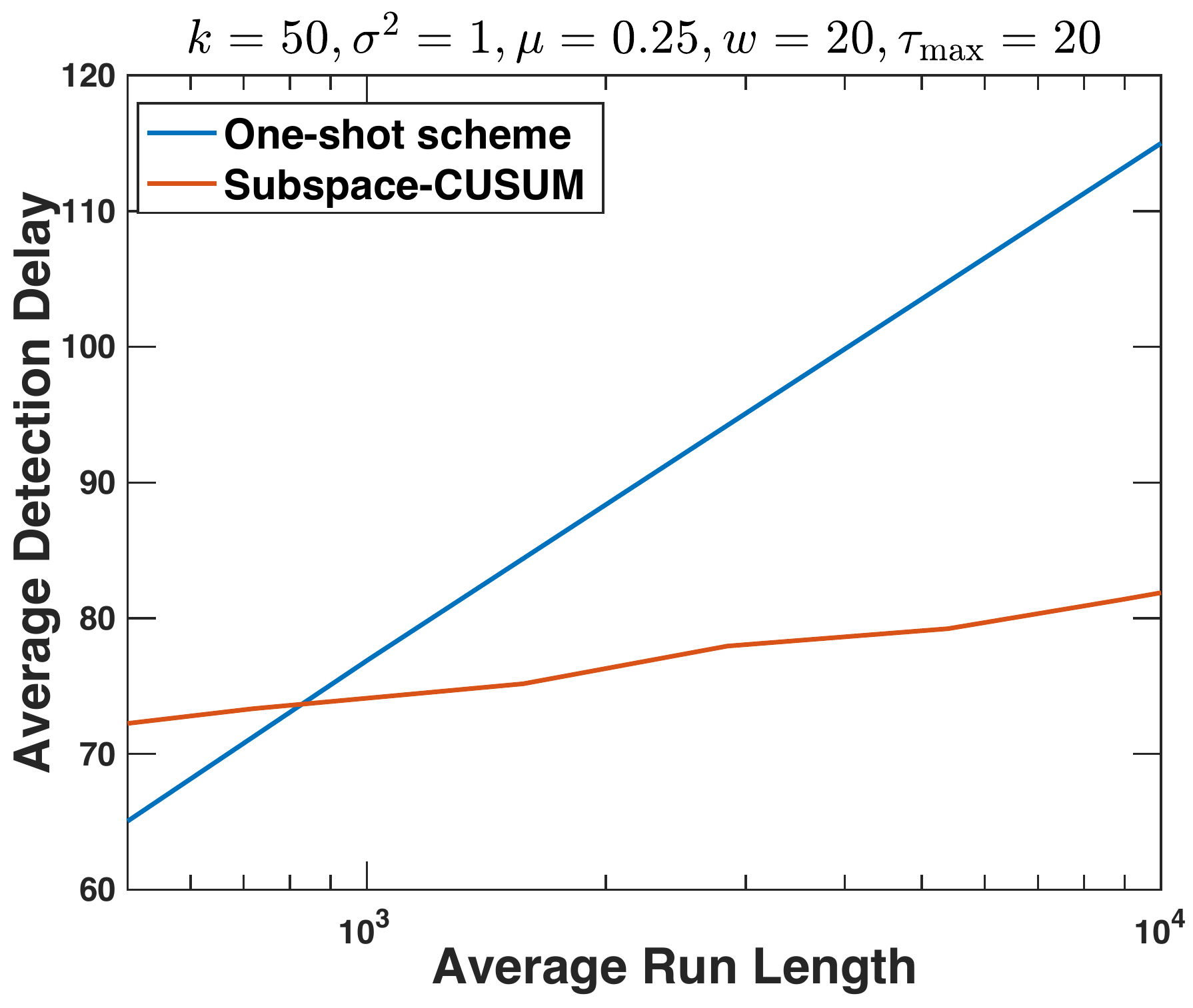}}
\vskip-0.2cm
\centerline{\footnotesize\hfil\hskip2.3cm (a)\hfill (b)\hskip1.7cm}
\vskip-0.2cm
\caption{Comparison of Average detection delay as a function of Average run length for Subspace-CUSUM and One-shot detection scheme.}
\vskip-0.2cm
\label{fig:compare}
\end{figure}
In Fig. \ref{fig:compare} we compare the average detection delay of the one-shot scheme (blue line) and our subspace-based test (red line) as a function of the average run length, which is the average period between false alarms. Fig. \ref{fig:compare} shows the results for (a) $\mu=0.1$ and (b) $\mu=0.25$, respectively. Our method can detect the weak signal with much smaller delay compared with the one-shot scheme. When $\mu$ is larger and since the one-shot scheme knows exactly $\mu$, for small average run length values it can outperform our method which does not know $\mu$ and the relative amplitudes $\alpha_i$ and estimates them. However the proposed method performs better when the average run length is larger. This suggests that combing multi-sensor observations can improve performance significantly.

\subsection{Seismic Data}

In this example, we consider the seismic tremor signal detection problem. When there is a tremor signal,  different seismic sensors will observe the same waveform with unknown and different delays. The tremor signals are useful for geophysical study and prediction of potential earthquakes. Usually, the tremor signals are very weak to detect using data at any individual sensor; therefore, network detection methods have been developed which essentially use covariance information of the data \cite{LiPeng2018}. This network-based detection problem can also be solved by our Subspace-CUSUM scheme discussed in \ref{sec:unknown}.

The seismic dataset we use is the records at Parkfield, California from 2am to 4am on 12/23/2004, which consist of 13 seismic sensors that simultaneously record a continuous stream of signals. The sampling frequency of the raw data is 250Hz. In the data preprocessing step, we normalize the observations at each sensor by subtracting the average value and dividing the maximal absolute value. The raw data after preprocessing is shown in Fig. \ref{fig:rawdata}. 

\begin{figure}[h]
\vskip-0.3cm
\centerline{
\includegraphics[width = 0.30\textwidth]{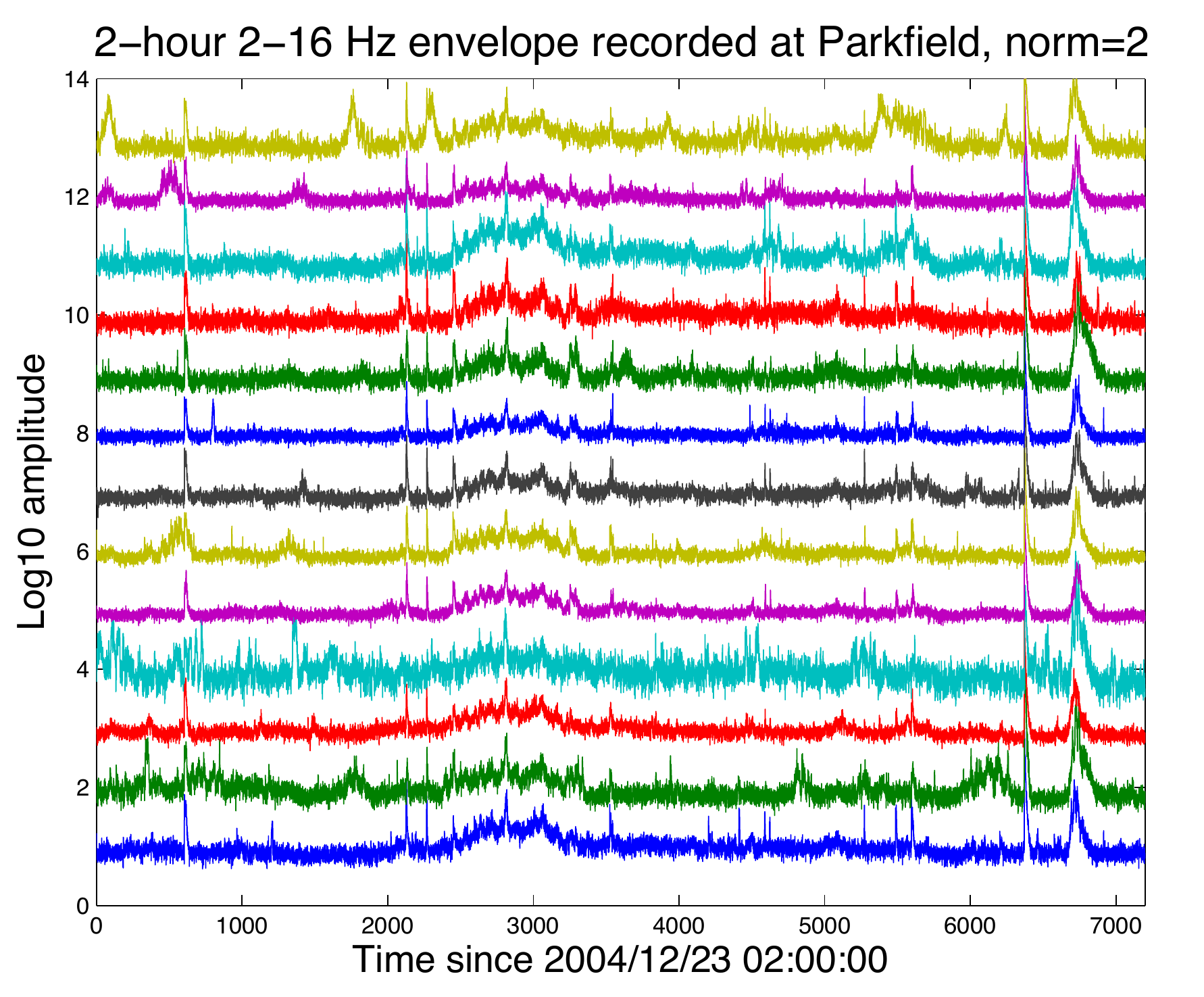}}
\vskip-0.2cm
\caption{Raw data from $13$ sensors.}
\vskip-0.2cm
\label{fig:rawdata}
\end{figure}

From the published catalog (Northern California Earthquake Data Center\footnote{\url{http://www.ncedc.org/ncedc/catalog-search.html}}), we see three small earthquakes as shown in table \ref{tab:earthquake}. There are also many low-frequency tremor records, mainly at time 2:34 $\sim$ 2:35, 2:42 $\sim$ 2:53, 3:24, 3:26, and 3:39. 
\begin{table}[h]
\vskip-0.3cm
\caption{Earthquake catalog at Parkfield during 2004/12/23 02-04 UT.}
\vskip-0.2cm
\centerline{
\begin{tabular}{c|c|c|c|c|c}
\hline
Date   &    Time      &        Lat    &    Lon   & Mag  &  Event ID \\
\hline
\!\!2004/12/23\! & 02:09:54.01 & 35.4593  & -120.7500  & 1.47    &   21429343\! \\
\!\!2004/12/23\! & 02:35:23.70  & 36.0368 & -120.6088 & 1.10  &   30229299\! \\
\!\!2004/12/23\! & 03:46:09.23  & 35.9290 & -120.4797  & 1.47   &    21429365\! \\
\hline
\end{tabular}}
\label{tab:earthquake}
\vskip-0.2cm
\end{table}
In this example, we assume that the maximum delay $\tau_{\max}=100$ namely $0.4$ seconds. The window size $w=200$, which corresponds to $0.8$\,sec. We use the data within the first $500$ sec (pre-change period) to find the drift parameter $d$ numerically, which is $1.5$ times the mean value of $(\hat{u}_t^\intercal \bm{x_t})^2$ in the first $500$ sec. We computed the Subspace-CUSUM statistic in \eqref{eq:detect_stat} which is shown in Fig. \ref{fig:seismic}. It can be seen that the three main peaks are at $603.6$\,sec, $2127.0$\,sec, $6370.0$\,sec respectively, which is quite close to the true earthquake time (recall that these times $594.0$\,sec, $2123.7$\,sec, and $6369.2$\,sec). There are some small and continuous peaks within the time period 2500\,sec $\sim$ 3200\,sec, which match the tremor catalog of 2:42 $\sim$ 2:53.
\begin{figure}[h]
\vskip-0.3cm
\centerline{\includegraphics[width = 0.24\textwidth]{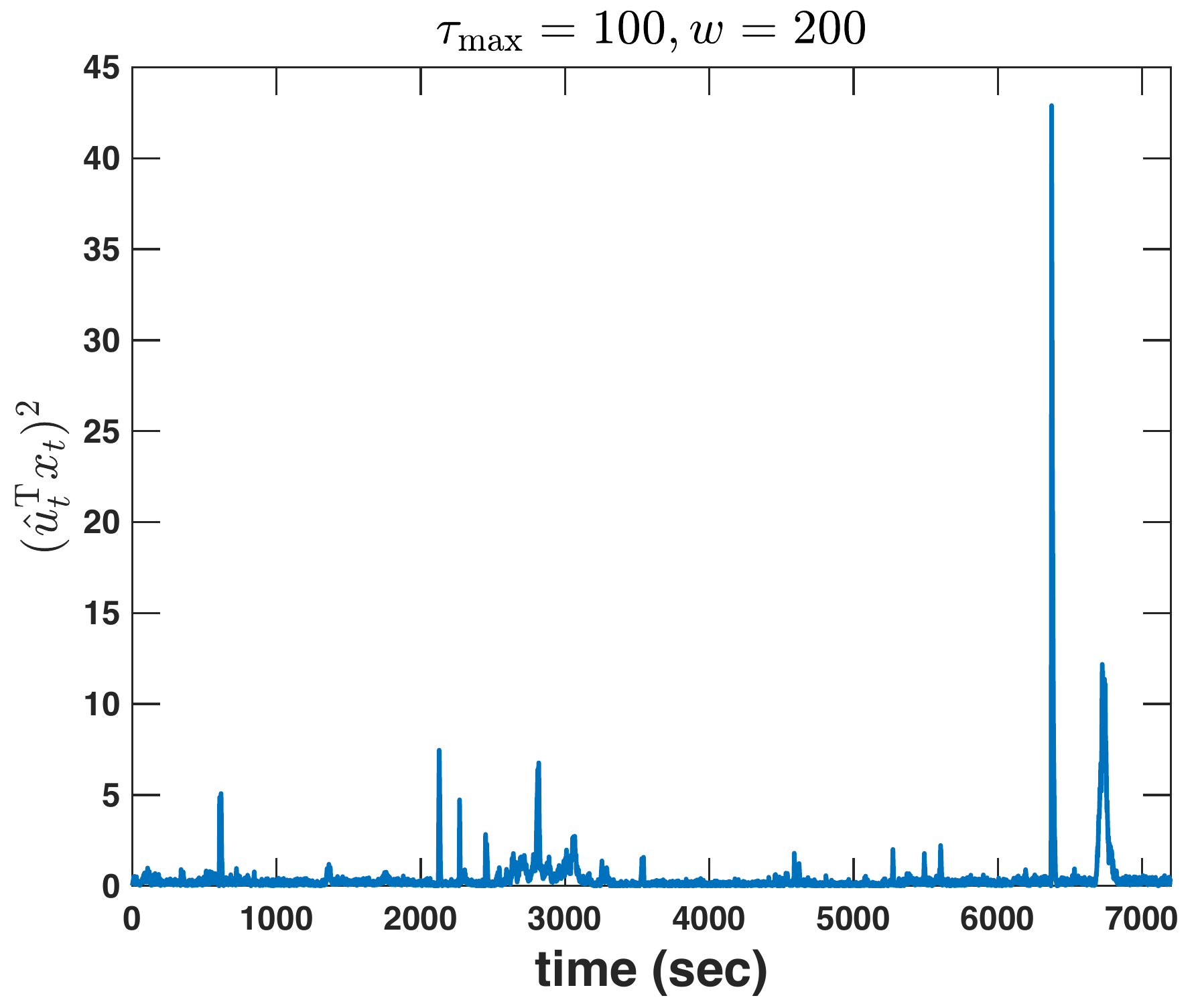} \hfill \includegraphics[width = 0.24\textwidth]{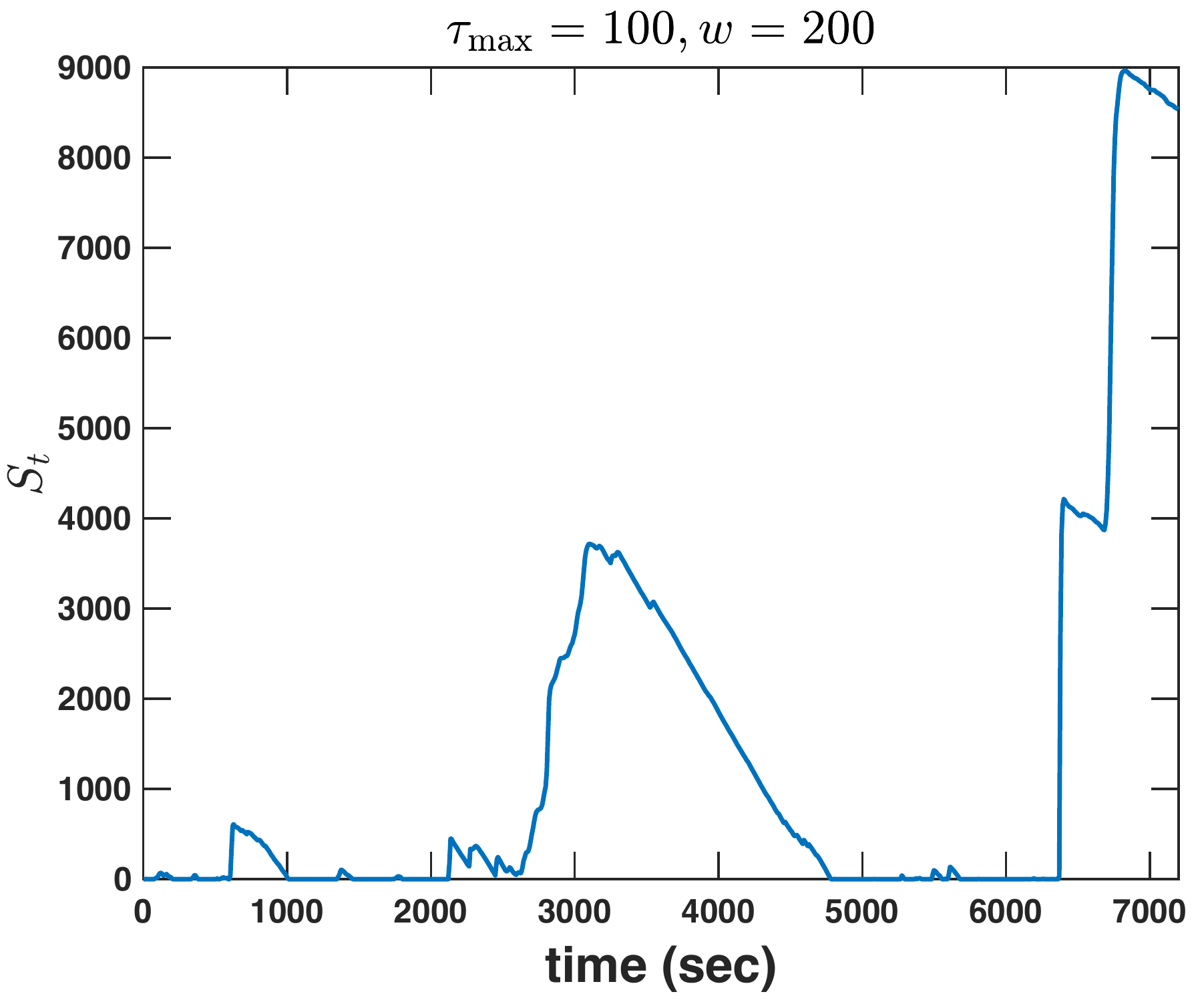} }
\vskip-0.2cm
\caption{Left: the increment $(\hat{u}_t^\intercal {\tilde x_t})^2$. Right: the Subspace-CUSUM detection statistics $\mathcal S_t$ over time.}
\vskip-0.2cm
\label{fig:seismic}
\end{figure}

\section{Conclusion}
We study the sequential change-point detection for asynchronous multi-sensors. We propose an asynchronous Subspace-CUSUM procedure based on delay estimation. Multi-sensor observations are synchronized and combined together to perform the change-point detection. We derive the optimal drift parameter for the proposed procedure, and characterize the relationship between the expected detection delay, average run length, and the energy of the time-varying signal. The good performance of the proposed procedure is presented using simulated signals and real seismic data. We also demonstrate that the proposed procedure outperforms the one-shot procedure in detecting weak and asynchronous signals.

\section*{Acknowledgment}
The work of George Moustakides was supported by the US National Science Foundation under Grant CIF\,1513373, through Rutgers University. The work of Liyan Xie and Yao Xie was supported by the US National Science Foundation under Grant CCF\,1442635, CMMI\,1538746, and a career award CCF\,1650913. The authors would like to thank Dr. Zhigang Peng at Georgia Institute of Technology on his helpful suggestions for seismic data processing.

\appendix
\section{Appendix}

\begin{proof}[Proof of Lemma \ref{lem:drift}]~
We follow the proof of Lemma 1 in \cite{xie2018first}, the only difference being that the SNR is now time varying. Based on Assumption \ref{ass:average_energy}, as $w \rightarrow \infty$, the average SNR over time equals to $\rho$. 

Similar to \cite{xie2018first}, let $\omega_t$ denotes the \textit{un-normalized} eigenvector and $v_t = \omega_t - u$ denotes the estimation error. Using Central Limit Theorem arguments \cite{anderson1963asymptotic, paul2007asymptotics}, we end up with
\begin{align*}
&\mathbb{E}_0[(\hat{u}_t^\intercal \bm x_t)^2] 
=\sigma^2+s^2(t)\|\alpha\|^2\mathbb{E}_0[(\hat{u}_t^\intercal u)^2]\\
&=\sigma^2+s^2(t)\|\alpha\|^2\mathbb{E}_0\left[\frac{1}{1+\|v_t\|^2}\right]\\
&\approx \sigma^2+s^2(t)\|\alpha\|^2\mathbb{E}_0[1-\|v_t\|^2] \\
&= \sigma^2 + s^2(t)\|\alpha\|^2\left( 1 -  \frac{1+\rho}{w\rho^2}(k-1)  \right).
\end{align*}
For the approximate equality we used the fact that to a first order approximation we can write $1/(1+\|v_t\|^2)\approx 1-\|v_t\|^2$ because $\|v_t\|^2$ is of the order of $1/w$ while the approximation error is of higher order. This completes the proof.
\end{proof}

\end{document}